\documentclass[reqno,11pt]{amsart}
\usepackage{amsmath,amssymb,tabularx,setspace,color}
\usepackage{pdflscape}
\newtheorem{theorem}{Theorem}

\newtheorem{definition}{Definition}
\newtheorem{remark}{Remark}
\newtheorem{corollary}{Corollary}
\usepackage{slashbox}
\usepackage{multirow}
\usepackage[bf,small,tableposition=top]{caption}
\usepackage{amsthm}
\usepackage{lineno}
\usepackage[figuresright]{rotating}
\usepackage{enumerate}
\usepackage[T1]{fontenc}
\usepackage[utf8]{inputenc}
\usepackage{graphics}
\usepackage{graphicx}
\usepackage{lscape}
\usepackage[dvipsnames]{xcolor}
\usepackage{mathrsfs}
\usepackage{hyperref}
\usepackage{textcomp}
\usepackage{listings}
\usepackage{amsmath}
\usepackage{lipsum}
\makeatletter
\renewcommand\section{\@startsection {section}{1}{\z@}%
                                   {-3.5ex \@plus -1ex \@minus -.2ex}%
                                   {2.3ex \@plus.2ex}%
                                   {\normalfont\large\bfseries}}
\makeatother

\begin{document}
\vspace{-1.9cm}
\doublespace
\title[]{\vspace{-1.9cm}Jackknife Empirical Likelihood-based  inference for S-Gini indices}
\author[]%
{ S\lowercase{reelakshmi} N$^{*}$, S\lowercase{udheesh} K K\lowercase{attumannil}$^{**,\dag}$ \lowercase{and} R\lowercase{ituparna} S\lowercase{en}$^{**}$\\
 $^{*}$I\lowercase{ndian} I\lowercase{nstitute of} T\lowercase{echnology},
  C\lowercase{hennai}, I\lowercase{ndia.}\\
$^{**}$ I\lowercase{ndian} S\lowercase{tatistical} I\lowercase{nstitute},
  C\lowercase{hennai}, I\lowercase{ndia.}\\
}
\thanks{ {$^{\dag}$} {Corresponding author E-mail: \tt \tt skkattu@isichennai.res.in. }}
\vspace{-0.9cm}
\doublespace
\begin{abstract}
Widely used  income inequality measure,  Gini index is extended to form  a family of income inequality measures known as Single-Series Gini (S-Gini) indices.  In this study, we develop empirical likelihood (EL) and  jackknife empirical likelihood (JEL) based inference for S-Gini indices. We prove that the limiting distribution of both EL and JEL ratio statistics are Chi-square distribution with one degree of freedom. Using the asymptotic distribution we construct EL and JEL based confidence intervals for realtive S-Gini indices.  We also give  bootstrap-t and  bootstrap calibrated empirical likelihood confidence intervals for S-Gini indices. A numerical  study is carried out to compare the performances of the proposed confidence interval with the bootstrap methods. A test for S-Gini indices  based on jackknife empirical likelihood ratio is also proposed.  Finally we illustrate the proposed method using an income data. 
\end{abstract}
\maketitle
\section{Introduction}
 Several indices of economic inequality, compatible with suitable  axioms, have been proposed
in the literature. For more than one century GMD and its derived measures (such as Gini index) celebrate a prominent role in the area of measurement of income inequality.  Gini mean difference is extended to form  generalized  families which vary in their properties and one such family is S-Gini family (Yitzhaki and Schechtman, 2013).  In this article, we discuss   statistical inference  associated with Single-Series Gini (S-Gini) family.  We refer to Donaldson and Weymark (1980), Yitzhaki(1983), Zitikis and Gastwirth (2002), Zitikis (2003)  and Barret and Donald (2009) and the references therein for the discussion on inference about S-Gini indices.



 \par Finding simple reliable estimators of different income  inequality measures and obtaining a consistent estimator for their asymptotic variance are important topic of research. Many authors discussed the estimation of income inequality as well as poverty  measures based on theory of U-statistics. For review of U-statistics based estimators see Formby et al. (2001) and Xu (2007). Xu (2000) explained the estimation of asymptotic variance of generalized Gini indices  using  iterated bootstrap method proposed by Hall (1992).    Zitikis (2003) obtained a plug in estimator  for S-Gini index and showed that the estimator is consistent and has asymptotic normal distribution. Giorgi et al. (2006) studied the asymptotic distribution of the plug-in estimators of  S-Gini indices and noted that bootstrap based confidence interval perform better than  normal approximation interval. Barret and Donald (2009) obtained an estimator of S-Gini  index and studied its asymptotic properties using influence function. Demuynck (2012) proposed an unbiased estimator for absolute S-Gini indices and  studied asymptotic  properties of the estimator using theory of combinatorics.

It is important to find the confidence interval for poverty and inequality measure to compare these measures in different population of interest. Empirical likelihood based confidence interval and likelihood ratio test (Owen 1988, 1990) received much attention recently. Qin et al. (2010) obtained an empirical likelihood confidence intervals for the Gini measure of income inequality and showed that the intervals based on normal or bootstrap approximation are less satisfactory than the bootstrap calibrated empirical likelihood ratio confidence intervals for  small or moderate sample size. Peng (2011) also independently discussed the empirical likelihood inference for Gini index and showed that the bootstrap calibration of the empirical likelihood method perform  better than the some other bootstrap methods.  Qin et al. (2013) discussed  empirical likelihood-based inferences for the Lorenz curve. They obtained the  profile empirical likelihood ratio statistics for the Lorenz ordinate under the simple as well as the stratified random sampling designs. Lv et  al. (2017) obtained a bootstrap-calibrated empirical likelihood confidence intervals for the difference between two Gini index.  In this work, first  we obtain empirical likelihood based confidence interval for relative S-Gini indices.

Implementation of empirical likelihood method is difficult when the maximization involve non-linear constraints. Motivated by this, Jing et al. (2009) proposed jackknife empirical likelihood (JEL) inference for obtaining confidence interval of a desired parametric function. They illustrated the JEL method using one as well as two sample U-statistics. Wang et al. (2016) proposed a jackknife empirical likelihood based confidence interval for the Gini index.  Wang and Shao (2016) derived the jackknife empirical likelihood for the difference of two Gini indices for dependent and independent data.  Recently, Lou and Qin (2018) obtained a kernel smoothing estimator for the Lorenz curve and developed  a smoothed jackknife empirical likelihood method for constructing confidence intervals of Lorenz ordinates. Sang et al. (2019) developed JEL based test for testing the equality of Gini correlation. In this work, we obtain a novel U-statistics estimator for S-Gini indices  which allows direct utilization of the jackknife empirical likelihood without involving any nuisance parameter.  %


 The rest of the article is organized as follows. In Section 2 we derive empirical log likelihood ratio statistic for relative S-Gini indices and prove that its limiting distribution is chi square distribution with one degree of freedom. In Section 3,  we obtain an  estimators for S-Gini indices based on U-statistics and study its asymptotic properties. Making use of this  we propose a jackknife empirical likelihood based confidence interval for relative S-Gini indices.  In Section 4, we report the result of a numerical study done to evaluate the performances of the proposed  confidence intervals. We also evaluate the performance of the jackknife empirical likelihood ratio test.  We illustrate our method using per capita personal income of the United States and the result is reported in Section 5. We conclude our study in Section 6.

\section{Empirical Likelihood Inference for Relative S-Gini indices}
In this section, we construct an empirical likelihood based confidence interval for relative S-Gini index. First we review the concept of Gini index and its variant. Let $X$ be a non-negative random variable with absolute continuous distribution function $F(.)$ and finite  mean $\mu=E(X)$.  Lorenz curve is defined as
\begin{equation}\label{lorenz}
 L\left( p \right) = \frac{1}{\mu }\int_0^{F^{-1}(p)} {tdF(t)},
\end{equation}
where  $p=F(x)$ and $F^{-1}(p)$ is the  $p$-{th} quantile of $X$.  The function $L$ is non-decreasing and convex which maps on to the interval $[0,1]$. Gini  index is defined as twice the area between Lorenz curve and the line of equality. It is given by
\begin{equation}\label{Ginil}
  G = 1 - 2\int_0^1 L(p)dp.
\end{equation}
Thus $G$ measures an extend to which the distribution of income among individuals within an economy
deviates from perfectly equal distribution. Gini index can be expressed in terms of covariance between $X$ and $F(X)$ as 
\begin{equation}\label{eq3}
  G = \frac{2}{\mu }Cov\left( {X,F(X)} \right).
\end{equation}
 Suppose the random variables $X_{1}$ and $X_{2}$ are distributed as $F$.   Gini mean difference (GMD) is defined as the expected absolute difference between  $X_1$ and $X_2$. That is
\begin{equation*}
  GMD = E|{X_1} - {X_2}| .
\end{equation*}Making use of the identity $|X_1-X_2|=2max(X_1,X_2)-X_1-X_2$, we can express GMD as
\begin{equation}
  GMD=4Cov\left(X,F(X)\right).
\end{equation}In view of (\ref{eq3}),  Gini index can be express as
\begin{equation}
  G = \frac{GMD}{2\mu }.
\end{equation}
\par Several income inequality measures are derived from GMD   by taking  different weights at the expectation and one among them is  S-Gini family of indices. We refer to  Yitzhaki  and Schechtman (2013) for more details about Gini based parameter.  The advantage of having  S-Gini family is that the evaluation of robustness of result can be done by knowing one member of that family (Barrett and Donald, 2009). The absolute and relative S-Gini indices are defined, as
\begin{equation}\label{agi}
S_\nu  =  - \nu Cov\left( {X,\bar F_X^{\nu  - 1}(X)} \right);\,\,\nu  > 0,\,\,\nu  \ne 1
\end{equation}
and
\begin{equation}\label{rgi}
R_\nu  = \frac{{ - \nu }}{\mu }Cov\left( {X,\bar F_X^{\nu  - 1}(X)} \right);\,\,\nu  > 0,\,\,\nu  \ne 1,
\end{equation} respectively, where $\bar{F}(x)=1-F(x)$ is the survival function of $X$ at $x$. Suppose $X_{(i)}$ denotes the $i$-{th} order statistic based on a random sample $X_1,X_2,\ldots,X_n$; from $F$. The plug-in estimator of absolute Gini indices is given by
\begin{equation}\label{estsgini}
\tilde S^\nu  = \frac{1}{n}\sum\limits_{i = 1}^n {X_i }  - \sum\limits_{i = 1}^n {\frac{{(n - i + 1)^\nu   - (n - i)^\nu  }}{{n^\nu  }}} X_{(i)}.
\end{equation}
Hence the plug-in estimator of relative S-Gini indices is given by
\begin{equation}\label{eq99}
  \tilde{R}^{\nu}  = 1 - \left[ \sum\limits_{i = 1}^n {X_i }  \right]^{ - 1} \sum\limits_{i = 1}^n {\frac{{(n - i + 1)^\nu   - (n - i)^\nu  }}{{n^{\nu-1}}}} X_{(i)}.
\end{equation}
We use above estimators to obtain the empirical likelihood based confidence interval for $R_\nu$.
Next,  we develop EL based confidence interval of relative S-Gini index.

Recalling the definition given in (\ref{rgi}), we have
\begin{eqnarray}\label{Sgi}
  R_\nu & =& \frac{{ - \nu }}{\mu }Cov\left( {X,\bar F_X^{\nu  - 1}(X)} \right)\nonumber \\
 &=&\frac{{ - \nu }}{\mu }\int_0^\infty {(x-\mu)\bar{F}^{\nu-1}(x)dF(x)}.
\end{eqnarray}
Hence  relative S-Gini index can be expressed as
\begin{equation}\label{rgempequ}
 R_\nu =\frac{E\left[(1-\nu\bar{F}^{\nu-1}(X))X\right]}{E(X)}.
 \end{equation} We use the identity given in (\ref{rgempequ}) to obtain the estimating equation that can be  used to construct empirical likelihood of $R_\nu$.
Using  a random sample  ${X_1,X_2,...,X_n}$; from $F$,  the empirical likelihood for $R_{\nu}$ is defined as
\begin{equation*}
  EL(R_\nu)=\sup_{\bf p} \left(\prod_{i=1}^{n}{p_i};\,\, \sum_{i=1}^{n}{p_i}=1;\,\,\sum_{i=1}^{n}{p_i C(X_i,R_\nu)}=0\right),
\end{equation*}
   where ${\bf p}=(p_1,p_2,...,p_n)$ is a probability vector and
\begin{equation*}
  C(X_i,R_\nu)=\left[1-\nu\bar{F}^{\nu-1}(X_i)\right]X_i-R_{\nu}X_i;\,\,i=1,2,...,n.
\end{equation*}
   Since above equation depends on unknown $\bar{F}(.)$,  we  replace  $\bar{F}(.)$ by $\bar{F}_{n}(.)$, the empirical survival function of $X$. Hence the profile empirical likelihood for $R_\nu$ is given by
   $$EL_{1}(R_\nu)=\sup_{\bf p}\left(\prod_{i=1}^{n}{p_i};\,\, \sum_{i=1}^{n}{p_i}=1;\,\,\sum_{i=1}^{n}{p_i \widehat{C}(X_i,R_\nu)}=0\right),$$
   where $$\widehat{C}(X_i,R_\nu)=\left[1-\nu\bar{F}^{\nu-1}_{n}(X_i)\right]X_i-R_\nu X_i;\,\,i=1,2,...,n.$$
   By Lagrange multiplier method, the maximum occurs at
   $$p_i=\frac{1}{n}\left(1+\lambda\widehat{C}(X_i,R_\nu)\right)^{-1}, \,\,i=1,2,...,n,$$  where $\lambda$ is the solution of
   $$\frac{1}{n}\sum_{i=1}^{n}{\frac{\widehat{C}(X_i,R_\nu)}{1+\lambda \widehat{C}(X_i,R_\nu)}}=0.$$
   Also note that, $\prod\limits_{k=1}^{n}p_i$, subject to $\sum\limits_{i=1}^{n}p_i=1$, attains its maximum $n^{-n}$ at $p_i=n^{-1}$. Hence,  the  empirical log likelihood ratio for $R_\nu$ is given by
   $$L(R_\nu)=2\sum_{i=1}^{n}\log\left[1+\lambda \widehat{C}(X_i,R_\nu)\right].$$
   The following theorem explains the limiting distribution of $L(R_\nu)$.
\begin{theorem}
 Let $h_{1}(x)=x{{\bar F}^{\nu - 1}}(x) +(\nu-1)\int_{0}^{x} {{{y \bar F}^{\nu  - 2}}(y)} d{F}(y)$ and assume that $E(h_1^{2}(X))<\infty$. As $n\rightarrow \infty$, the distribution of $L(R_\nu)$ is a scaled chi-square distribution with one degree of freedom.
That is,

$$L(R_\nu) \xrightarrow[]{d} \frac{\sigma_2^{2}}{\sigma_1^{2}}\chi^{2}(1),$$
where
\begin{equation*}\label{e12}
  \sigma_1^{2}=Var[(1-\nu \bar{F}^{\nu - 1}(X)-R_\nu) X]
\end{equation*}
and

\begin{equation*}\label{e13}
  \sigma_2^{2}=Var[(1-2h_{1}(X)-R_\nu)X].
\end{equation*}
\end{theorem}
\begin{proof}
Using the distribution function of $\min(X_1,X_2,\ldots,X_v)$, it is easy to verify $E(h_1(X))=E(\nu X\bar{F}(X))$. Consider
\begin{equation*}
  \frac{1}{\sqrt{n}}\sum_{i=1}^{n}C(X_i,R_\nu)=\frac{1}{\sqrt{n}}\sum_{i=1}^{n}\big((1-2h_{1}(X_i)-R_\nu)X_i+E(h_{1}(X))\big)+o_p(1).\end{equation*}
Therefore by central limit theorem, as $n \rightarrow \infty$
\begin{equation*}
  \frac{1}{\sqrt{n}}\sum_{i=1}^{n}C(X_i,R_\nu)\xrightarrow[]{d}N(0,\sigma_2^{2}).
\end{equation*}
Since $E\big((1-\nu \bar{F}(X)-R_\nu )X\big)=0$, we have $\sigma_1^{2}=E\big((1-\nu \bar{F}(X)-R_\nu )X\big)^2$. By law of large number, as $n\rightarrow \infty$
$$\frac{1}{n}\sum_{i=1}^{n}C^{2}(X_i,R_\nu)=\frac{1}{n}\sum_{i=1}^{n}\big((1-\nu \bar{F}_{n}(X_i)-R_\nu )X_i\big)^2=\sigma_1^{2}+o_p(1).$$
Therefore, by using Slutsky's theorem, as $n\rightarrow \infty$, the empirical log likelihood ratio
\begin{eqnarray*}
L(R_\nu)&=&2\sum_{i=1}^{n}\log\left[1+\lambda \widehat{C}(X_i,R_\nu)\right]\\&=&\frac{\left[\frac{1}{\sqrt{n}}\sum_{i=1}^{n}\widehat{C}(X_i,R_\nu)\right]^2}{\frac{1}{n}
\sum_{i=1}^{n}\widehat{C}^{2}(X_i,R_\nu)}+o_{p}(1)\xrightarrow[]{d}\frac{\sigma_2^{2}}{\sigma_1^{2}}\chi^2(1).
\end{eqnarray*}
\end{proof}
Using the asymptotic distribution of empirical log likelihood ratio,  we can construct  EL based confidence interval for relative S-Gini indices. Let
$\widehat{\sigma}_1^{2}$ and $\widehat{\sigma}_2^{2}$ be the plug in estimators of ${\sigma}_1^{2}$ and ${\sigma}_2^{2}$, respectively.  For $0<\alpha<1$, a $(1-\alpha)$ level empirical likelihood based confidence interval for $R_{\nu}$ can be obtained as
$$\left(R_\nu:L(R_\nu)\le \frac{\widehat{\sigma}_2^{2}}{\widehat{\sigma}_1^{2}}\chi^2_{1-\alpha}(1)\right),$$
where $\chi^2_{1-\alpha}(1)$ is the upper $\alpha$ percentile of chi-square distribution with one degree of freedom.

\section{JEL based inference of relative S-Gini indices}
\par  The empirical likelihood can be implemented easily when we are maximizing a non-parametric likelihood subject to a set of linear constrains. However,  maximization  involving nonlinear constrains are computationally  difficult.  For example, if an estimator is a U-statistic with a kernel of degree 2 we need to consider constrains in quadratic form to maximize the non-parametric likelihood. In this scenario JEL is better alternative for empirical likelihood.  In this section, we develop jackknife empirical likelihood based confidence interval for relative S-Gini indices. We also develop a jackknife empirical likelihood ratio test for testing $R_\nu=R_0$, where $R_0$ is a real number belongs to the interval $[0,1]$.  Accordingly, first  we obtain an estimator of $R_\nu$ and discuss its properties.

Using equation (\ref{agi}),  we can express absolute S-Gini index   as
\[S_\nu  = \mu  - E\left[ {\min \left( {{X_1},{X_2},{X_3},...,{X_\nu }} \right)} \right],\]
provided $\nu$ is an integer. Suppose  $X_{1},X_{2},...,X_{n}$ are $n(\geq \nu)$ independent and identically distributed samples from $F$. Define a symmetric kernel $h(.)$ of degree $\nu$ as
\begin{equation}\label{kernel}
h\left( {{X_1},{X_2},...,{X_\nu }} \right) = \frac{{{X_1} + {X_2}+ ... + {X_\nu } - \nu \min \left( {{X_1},{X_2},...,{X_\nu }} \right)}}{\nu }.
\end{equation}An unbiased estimator of absolute S-Gini index  based on U-statistic  is given by
\begin{equation}\label{use}
\widehat S_\nu  = \frac{1}{\binom{n}{\nu}}\sum_{\binom{n}{\nu}} {h\left( {{X_{i1}},{X_{i2}},...,{X_{i\nu} }} \right),}
\end{equation}
where the summation is over the set $\binom{n}{\nu}$ of all combinations of $\nu$ integers,
$i_1<i_2<...<i_{\nu}$ chosen from $(1,2,...n)$. When $\nu=2$, in terms of order statistics we have the following equivalent expression
\begin{equation*}
\sum\limits_{i = 1}^n \sum\limits_{j = 1,j<i}^n\min\{X_i,X_j\} =\sum\limits_{i = 1}^n {(n - i)X_{(i)} }.
\end{equation*}
 And for $\nu$=3, we obtain
\begin{eqnarray*}
\sum\limits_{i = 1}^n \sum\limits_{j= 1,j<i}^n\sum\limits_{k = 1,k<j}^n\min\{X_i,X_j,X_k\} &=&\sum\limits_{i = 1}^n {\frac{{(n - i - 1)(n - i)X_{(i)} }}{2}}  \\&=& \sum\limits_{i = 1}^n {\binom{n - i}{2} } X_{(i)}.
\end{eqnarray*}
In general,  the  estimator of absolute S-Gini index given in (\ref{use}) can be expressed  as
\begin{equation}\label{uses}
\widehat S_\nu   = \frac{1}{{\binom{n}{\nu}  }}\left\{ \binom{n - 1}{\nu  - 1} \sum\limits_{i = 1}^n \frac{{{X_i } }}{\nu } - \sum\limits_{i = 1}^n {\binom{n - i}{\nu  - 1} X_{(i)} } \right\}.
\end{equation}
 Denote $\bar X= \frac{1}{n}\sum\limits_{i = 1}^n {X_i}$.  Hence an estimator of relative S-Gini index is the ratio of two U-statistics given by
\begin{equation}\label{usergi}
\widehat{R}_\nu  = \frac{{\widehat S_\nu }}{\bar X}.
\end{equation}


Next we prove the asymptotic properties of the estimators given in (\ref{uses}) and (\ref{usergi}) which we use to prove the limiting distribution of JEL ratio statistic. First we prove the consistency of the estimators (\ref{uses}) and (\ref{usergi}). Since $\widehat S_\nu $ is a U-statistic, as $n \rightarrow \infty$,  $\widehat S_\nu$  converges in probability to $S_\nu$ (Lehmann, 1951). 
\begin{theorem}
As $n\rightarrow \infty$, $\widehat {R}_\nu$  converges in probability to $R_\nu$.
\end{theorem}

\begin{proof}
By law of large numbers, as $n\rightarrow \infty$, ${\bar X}$ converges in probability to $\mu$. Since the estimator $\widehat{R}_\nu $ can be written as
\[\widehat{R}_\nu = \frac{{\widehat S_\nu}}{{S_\nu}}\frac{\mu }{{\bar X}}\frac{{S_\nu}}{\mu },\]
we have the result.
\end{proof}

\noindent Next we obtain the asymptotic distribution of the estimators $\widehat S_\nu$ and $\widehat{R}_\nu$.
\begin{theorem}
 As $n \to \infty$, the distribution of $\sqrt n \left( {\widehat S_\nu - S_\nu} \right)$ is Gaussian with mean zero and variance $\sigma^2$ where $\sigma^2$ is given by
\[\sigma^2 = V\Big( {{X}(1 - \nu {{\bar F}^{\nu - 1}}(X)) -\nu(\nu-1)\int_{0}^{X} {{{y \bar F}^{\nu  - 2}}(y)} d{F}(y)} \Big).\]
\end{theorem}
\begin{proof}
The asymptotic normality of $\widehat{S}_\nu $  can be proved using central limit theorem  for U-statistics.  The asymptotic variance is $\nu^2\sigma_{3}^2$ (Hoeffding, 1948), where
\begin{equation}\label{avar}
 \sigma_{3}^{2}=Var\left(E\big(h\left( {{X_1},{X_2},...,{X_\nu }} \right)|X_1=x\big)\right).
\end{equation}
Denote $Z=\min(X_2,X_3,...,X_\nu)$, then the distribution of $Z$ is given by $1-\bar{F}^{\nu-1}(x)$, where $\bar{F}(x)=1-F(x)$.  Consider
\begin{eqnarray*}\label{eq4}
E\left[ {\min \left( {x,{X_2},{X_3},...,{X_\nu }} \right)} \right]& =& E\left[ {xI(Z > x)} \right] + E\left[ {ZI(Z \le x)} \right]\nonumber\\&=&x{{\bar F}^{\nu - 1}}(x) + (\nu-1)\int_0^x {{{y \bar F}^{\nu  - 2}}(y)} d{F}(y).
\end{eqnarray*}
Using (\ref{kernel}), we have
\begin{small}
\begin{equation*}
 E\left(h\big( {{X_1},{X_2},...,{X_\nu }}|X_1=x\big) \right)=\frac{1}{\nu}\Big(x(1-\nu{\bar F}^{\nu - 1}(x)) - \nu(\nu-1)\int_0^x {{{y \bar F}^{\nu  - 2}}(y)} d{F}(y)\Big).
\end{equation*}
\end{small}
\noindent Hence , from (\ref{avar}) we obtain the variance expression specified in the theorem.
\end{proof}

\noindent Note that, as $n \to \infty$, ${\bar X}$ converges in probability to $\mu$. Hence by Slutsky's theorem,  from Theorem 2  we have the asymptotic normality of $\widehat {R}_\nu$  and  we state it as next result.
\begin{corollary}
 As $n \to \infty$, the distribution of $\sqrt n \left( {\widehat {R}_\nu - R_\nu} \right)$ is Gaussian with mean zero and variance $\frac{{ \sigma^{2}}}{{{\mu ^2}}}.$
\end{corollary}
Next, we discuss the construction of jackknife empirical likelihood ratio for $R_\nu$. Let
 \begin{equation}\label{kerneljel}
 \tilde{h}(X_{1},X_{2},...,X_{\nu};\,R_\nu)=\frac{1}{\nu}(X_{1}+X_{2}+...+X_{\nu})R_\nu-h(X_{1},X_{2},...,X_{\nu})
\end{equation}
where $h(X_{1},X_{2},...,X_{\nu})$ is given in (\ref{kernel}). Since $E({h}(X_{1},X_{2},...,X_{\nu};)=R_\nu E(X)$ we have  $E(\tilde{h}(X_{1},X_{2},...,X_{\nu};\,R_\nu))=0$.
Define new estimating equation for $R_v$ as
\begin{equation}\label{JELest}
  \tilde R_\nu=\frac{1}{{\binom{n}{\nu}}}\sum_{\binom{n}{\nu}}^{}\tilde{h}(X_{1},X_{2},...,X_{\nu};R_\nu)=0.
\end{equation}The importance of the equation (\ref{JELest}) is that we can study asymptotic properties of jackknife empirical likelihood under the framework developed by Jing et al. (2009).
The jackknife pseudo values for $R_\nu$  are given by
\begin{equation*}
  \widehat{V}_k=n\tilde{R}_\nu-(n-1)\tilde{R}_{\nu,k};\,\,k=1,2,...,n,
\end{equation*}
where $\tilde{R}_{\nu,k}$, $k=1,2,...,n$ can be obtained  from (\ref{JELest}) using $(n-1)$ observations $X_1,X_2,...,X_{k-1},$ $X_{k+1},...,X_{n}$. The JEL for $R_\nu$ is defined as
\begin{equation}\label{J1}
JEL(R_\nu)=\sup_{\bf p} \left(\prod_{k=1}^{n}{p_k};\,\, \sum_{k=1}^{n}{p_k}=1;\,\,\sum_{k=1}^{n}{p_k \widehat{V}_k}=0\right).
\end{equation}
The maximum of \eqref{J1} occurs at
$p_k=\frac{1}{n}\left(1+\lambda_{1}\widehat{V}_k\right)^{-1}, \,k=1,2,...,n$  where $\lambda_1$ is the solution of
  \begin{equation}\label{lambda}
\frac{1}{n}\sum_{k=1}^{n}{\frac{\widehat{V}_k}{1+\lambda_1 \widehat{V}_k}}=0,
  \end{equation}
   provided
\begin{equation}\label{lamb11}
  \min_{{1\le k\le n}}\widehat{V}_{k}<\widehat R_\nu<  \max_{1\le k\le n}\widehat{V}_{k}.
\end{equation}
Also note that, $\prod\limits_{k=1}^{n}p_i$, subject to $\sum\limits_{i=1}^{n}p_i=1$, attains its maximum $n^{-n}$ at $p_i=n^{-1}$. Hence, the jackknife empirical log-likelihood ratio  for $R_\nu$ is given by
\begin{equation}\label{jelrat}
  J(R_{\nu})=2\sum_{i=1}^{n}\log\left[1+\lambda_1 \widehat{V}_k\right].
\end{equation}
  To find the JEL based confidence interval, we need to find the limiting distribution of jackknife empirical log-likelihood ratio statistic and the result is stated in the following theorem.
  \begin{theorem}
  Let $g(x)=E\left(\tilde{h}(X_1,X_2,...,X_{\nu};R_\nu)|X_1=x\right)$ and  assume that $E\left(\tilde{h}^2(X_1,...,X_\nu;R_\nu)\right)<\infty$  and $\sigma_{g}^{2}=\nu^2Var(g(X_{1}))>0$. Then, as $n\rightarrow \infty$
\begin{equation*}
  J(R_\nu)\xrightarrow[]{d} \chi^{2}(1).
\end{equation*}
  \end{theorem}
  \begin{proof}
  Let $S^2=\frac{1}{n}\sum_{k=1}^{n}\widehat{V}_{k}^{2}$.
    Since $\tilde R_\nu=\frac{1}{n}\sum_{k=1}^{n}\widehat{V}_k$, by strong law of large number we have
    \begin{equation}\label{eq8}
      S^2= \sigma_{g}^{2}+o(1).
      \end{equation}
Using Lemma A.4 of Jing et al. (2009) we have
      \begin{equation}\label{eq9}
    \max_{1\le k\le n} |\widehat{V}_{k}|=o(\sqrt{n}).
  \end{equation}
Above two equations yield
\begin{equation}\label{eq111}
  \frac{1}{n}\sum_{k=1}^{n}|\widehat{V}_{k}|^{3}\le  |\widehat{V}_{k}|\frac{1}{n}\sum_{k=1}^{n}\widehat{V}_k^{2}=o(\sqrt{n}).
\end{equation}
The $\lambda_1$ satisfies the equation (\ref{lambda}) has the property  (Jing et al., 2009)
\begin{equation}\label{eq25}
  |\lambda|=O_{p}(n^{-\frac{1}{2}}).
\end{equation}
Hence using (\ref{eq9}) we have
\begin{equation}\label{eq10}
    \max_{1\le k\le n} \lambda|\widehat{V}_{k}|=o({1}).
  \end{equation}Hence
  \begin{equation*}
  \frac{1}{n}\sum_{k=1}^{n}\widehat{V}_k^{3}\lambda^2|1+\lambda \widehat V_k|^{-1}=o(\sqrt{n})O_p({1/n})o(1)=o_p(1/\sqrt{n}).
\end{equation*} Since $\tilde R_\nu=\frac{1}{n}\sum_{k=1}^{n}\widehat{V}_k$, from (\ref{lambda}),  we obtain
\begin{equation}\label{eq11}
\lambda=\frac{\tilde R_\nu}{S^2}+o_p(1/\sqrt n).
\end{equation}
Using Taylor's theorem, we can express  $J(R_\nu)$ given in (\ref{jelrat}) as
\begin{equation}\label{asylike}
  J(R_\nu)=2n\lambda\tilde R_\nu-n S^2\lambda^2+Rim(R_\nu),
\end{equation}where $Rim(R_\nu)$ is the reminder term.  Using $|\lambda|=O_{p}(n^{-\frac{1}{2}})$ and (\ref{eq111}) it is easy to verify  that the reminder term $Rim(R_\nu)$ is $o_p(1)$.  Hence using (\ref{eq11}), the expression in (\ref{asylike}) can be written as
\begin{equation}\label{eq15}
  J(R_\nu)=\frac{n\tilde R_\nu^2}{S^2}+o_p(1).
\end{equation}
Using the central limit theorem for U-statistics, as $n \to \infty$, $\sqrt{n}\tilde R_\nu$ converges in distribution to normal with mean zero and variance $\sigma_{g}^{2}$.  Accordingly $\frac{n\tilde R_\nu^2}{\sigma_{g}^{2}}$  converges in distribution to $ \chi^{2}$ with one degree of freedom.   Since $S^2= \sigma_{g}^{2}+o(1)$ by Slutsky's theorem, as $n\rightarrow \infty$, $J(R_\nu)$ converges in distribution to $ \chi^{2}$ with one degree of freedom.
    \end{proof}
  Using Theorem 4, we can constructed a $100(1-\alpha) \%$ JEL based confidence interval for $R_\nu$ as
\begin{equation*}
  \left(R_\nu|J(R_\nu)\le \chi^2_{1-\alpha}(1)\right),
\end{equation*}where $\chi^2_{(1-\alpha)}(1)$ is the $(1-\alpha)-$th percentile of chi-square distribution with one degree of freedom. The performance of these confidence intervals are evaluated through Monte carlo simulation and the results are reported in Section 4.

Using the asymptotic distribution of jackknife empirical  log likelihood ratio we can develop JEL based test for testing the hypothesis $R_\nu=R_0$, where $R_0$ is a specific value of $R_\nu$.  We reject the hypothesis if $$J(R_\nu)> \chi^2_{1,1-\alpha},$$ where $\alpha$ is the desired significance level.  Simulation study shows that the type 1 error rate of the test converges to desired significance level and has very good power for diffrent alternatives.  The results of the related simulation study are also reported in Section 4.
\section{Simulation results}
The proposed JEL based confidence interval and test are evaluated through numerical study. We compare the JEL based confidence interval with bootstrap based confidence intervals and the performance of these  confidence  intervals are compared in terms of coverage probability and  average length. To evaluate the JEL based test, we find the empirical type 1 error and the power of the test. The simulation is done using R and  repeated for thousand times.

First, we investigate performances of the  confidence intervals based on  bootstrap-t (Boot\_t), bootstrap calibrated empirical likelihood (BCEL) and JEL methods.   For comparison, we consider the BCEL confidence intervals over empirical likelihood  ratio confidence intervals as it suffers from under coverage problems for small sample sizes. We considered thousand  bootstrap replicates to obtain Boot\_t and BCEL confidence intervals.

 Next, we summarize the procedures for the construction of BCEL confidence interval. The algorithm is given below.
\begin{enumerate}
\item For each bootstrap sample, indexed by $b=1,2\ldots,B$, draw a bootstrap sample $({X_{1}^b},{X_{2}^b},...,{X_{n}^b})$ with replacement from the original random  sample ${X_{1}},{X_{2}},...,{X_{n}}$; from $F$.
\item Calculate the empirical log likelihood ratio  $$L^b(R_\nu)=2\sum_{i=1}^{n}\log\left[1+\lambda^b \widehat{C}(X_i^b,R_\nu)\right],$$ where
    \begin{equation*}
      \widehat{C}(X_i^b,R_\nu)=\left[1-\nu\bar{F}^{\nu-1}_{n}(X_i^b)\right]X_i^b-R_{\nu}X_i^b;\,\,i=1,2,...,n
    \end{equation*}
         and $\lambda^b$ is the solution of  $$\frac{1}{n}\sum_{i=1}^{n}{\frac{\widehat{C}(X_i^b,R_\nu)}{1+\lambda^b \widehat{C}(X_i^b,R_\nu)}}=0.$$
\item Find  $ I_{\alpha}$, the upper $100\alpha\%$ sample quantile of
$L^{1}(R_\nu), L^{2}(R_\nu),..., L^{B}(R_\nu)$.

 \item A $(1-\alpha)$ level BCEL  confidence  interval of $R_\nu$ is given by
\begin{equation*}
  \left(R_\nu|L(R_\nu)\le I_{\alpha}\right).
\end{equation*}
\end{enumerate}

\par Next, we discuss discuss the algorithm for  obtaining  bootstrap\_t confidence interval. Using a random sample  $X_{1},X_{2}...,X_{n}$;  from $F$, a $100(1-\alpha)\%$   bootstrap-t confidence interval is
\begin{equation*}
 \left(\widehat{R}_{\nu}-T_{1-\alpha/2} \widehat se(\widehat{R}_{\nu}),\widehat{R}_{\nu}-T_{\alpha/2}\widehat se(\widehat{R}_{\nu})\right),
\end{equation*}where $T_{1-\alpha/2}$, $T_{\alpha/2}$ and $ \widehat se(\widehat{R}_{\nu})$ can be computed as outlined below.

\begin{enumerate}
\item For each bootstrap sample, indexed by $b=1,2\ldots,B$, draw a sample $({X_{1}^b},{X_{2}^b},...,{X_{n}^b})$ with replacement from $({X_{1}},{X_{2}},...,{X_{n}})$.
\item  Compute $\widehat{R}_{\nu}^b$ from the $b-$th bootstrap sample.
\item  Compute  $ \widehat se(\widehat{R}_{\nu})$, the sample standard deviation of the replicates $\widehat{R}_{\nu}^b.$
\item Compute $T^b=\frac{\widehat{R}_{\nu}^{b}-\widehat{R}_{\nu}}{\widehat se( \tilde R_{\nu}^b)}$, $b=1,2\ldots,B$. To find $\widehat se( \tilde R_{\nu}^b)$ one need to obtain further bootstrap sample $({X_{1}^{*b}},{X_{2}^{*b}},...,{X_{n}^{*b}})$ from  $({X_{1}^b},{X_{2}^b},...,{X_{n}^b})$.
    \item Find  the $(\alpha/2)-$th  and $(1-\alpha/2)-$th  sample quantiles $T_{\alpha/2}$ and $T_{1-\alpha/2}$, from the ordered sample of replicates $T^b$.
    \end{enumerate}


 \begin{table}[h]
\centering
\caption{Exponential distribution ($\lambda=1$)}
\label{my-label}
\begin{tabular}{|l|l|l|l|}
\hline
n                   & Interval & Coverage probability & Average length \\ \hline
\multirow{3}{*}{20} & Boot\_t    &  91.97                &   0.3397             \\ \cline{2-4}
                    & BCEL     &  93.77                    &   0.3149             \\ \cline{2-4}
                    & JEL      &  94.18                & 0.3046               \\ \hline
\multirow{3}{*}{40} & Boot\_t    &  92.12                & 0.3223                \\ \cline{2-4}
                    & BCEL     &  93.97                    &  0.2590              \\ \cline{2-4}
                    & JEL      &  94.42                &  0.2016              \\ \hline
\multirow{3}{*}{60} & Boot\_t    &  93.60                & 0.2163               \\ \cline{2-4}
                    & BCEL     &  92.35                    & 0.1989               \\ \cline{2-4}
                    & JEL      & 94.19                 & 0.1642               \\ \hline
\multirow{3}{*}{80} & Boot\_t    &94.52                   & 0.1931               \\ \cline{2-4}
                    & BCEL     & 94.28                     &  0.1426              \\ \cline{2-4}
                    & JEL      &94.94                   &  0.1391              \\ \hline
\end{tabular}
\end{table}

\begin{table}[h]
\centering
\caption{Pareto distribution ($\alpha=10, k=1$)}
\label{my-label}
\begin{tabular}{|l|l|l|l|}
\hline
n                   & Interval & Coverage probability & Average length \\ \hline
\multirow{3}{*}{20} & Boot\_t    &  90.82                &0.1102                \\ \cline{2-4}
                    & BCEL     &  91.71                    & 0.0904               \\ \cline{2-4}
                    & JEL      & 92.70                 &  0.0844              \\ \hline
\multirow{3}{*}{40} & Boot\_t    & 92.30                 & 0.0768                 \\ \cline{2-4}
                    & BCEL     &   94.00                   & 0.0589              \\ \cline{2-4}
                    & JEL      & 94.93                 & 0.0559               \\ \hline
\multirow{3}{*}{60} & Boot\_t    & 92.76                 &  0.0683              \\ \cline{2-4}
                    & BCEL     & 95.21                     &  0.0725              \\ \cline{2-4}
                    & JEL      & 94.18                 & 0.0455               \\ \hline
\multirow{3}{*}{80} & Boot\_t    & 95.70                 &  0.0469              \\ \cline{2-4}
                    & BCEL     &  95.81                    &  0.0640              \\ \cline{2-4}
                    & JEL      &94.88                  & 0.0337               \\ \hline
\end{tabular}
\end{table}
\begin{table}[]
\centering
\caption{Log normal distribution ($\mu=0, \sigma^2=1$)}
\label{my-label}
\begin{tabular}{|l|l|l|l|}
\hline
n                   & Interval & Coverage probability & Average length \\ \hline
\multirow{3}{*}{20} & Boot\_t    &  91.80               &0.1102                \\ \cline{2-4}
                    & BCEL     &  92.34                    & 0.0904               \\ \cline{2-4}
                    & JEL      & 92.82                 &  0.0844              \\ \hline
\multirow{3}{*}{40} & Boot\_t    & 92.97                 & 0.0768                 \\ \cline{2-4}
                    & BCEL     &   94.83                   & 0.0589              \\ \cline{2-4}
                    & JEL      & 94.96                 & 0.0559               \\ \hline
\multirow{3}{*}{60} & Boot\_t    & 93.70                 &  0.0683              \\ \cline{2-4}
                    & BCEL     & 94.92                     &  0.0725              \\ \cline{2-4}
                    & JEL      & 94.85                 & 0.0455               \\ \hline
\multirow{3}{*}{80} & Boot\_t    & 94.90                 &  0.0469              \\ \cline{2-4}
                    & BCEL     &  95.01                    &  0.0640              \\ \cline{2-4}
                    & JEL      &94.98                  & 0.0337               \\ \hline
\end{tabular}
\end{table}
 First, we simulate observations from unit exponential where the true value of $R_3$ is $0.67$. We find 95\% confidence intervals for relative S-Gini index using all three methods discussed above. The coverage probability and average length obtained for different sample sizes are reported in Table 1.  Next, we find the confidence intervals for $R_\nu$ when observations are generated from Pareto distribution with survival function $\bar F (x)=(\frac{k}{x})^{\alpha};\,x>k$. When $k=1$ and $a=10$, the true value of $R_3$ is $0.068$. The coverage probability and average length of the confidence intervals obtained for $R_\nu$ correspond to Pareto case are reported in Table 2.

\par When the sample size increases, Boot\_t and BCEL are comparable for the exponential distribution, but these show some over coverage problems for Pareto distribution. In almost all cases, Boot\_t has wider length compared to BCEL except for Pareto distribution when $n=60$. For small samples, JEL performs better than Boot\_t as well as BCEL in terms of average length. Giorgi et al. (2006) explained the superiority  of Boot\_t intervals over normal approximation based intervals for relative  S-Gini indices. Qin et al. (2010) discussed the  performance of  bootstrap calibrated empirical likelihood intervals  over Boot\_t intervals for Gini index. In our simulation study, in most of the cases, the jackknife empirical likelihood confidence interval performs better  than the Boot\_t and the bootstrap calibrated empirical likelihood  confidence intervals for relative S-Gini indices.

Finally, we generated observations form log normal distribution with parameter $\mu=0$ and $\sigma^2=1$. The true value of $R_3$ is  $0.660$ and the result of the simulation study is reported in Table 3. From Table 3 it is clear that the JEL intervals has better coverage probability and shorter length than the Boot\_t and  BCEL intervals. For large sample size, even though the coverage probabilities of  JEL and BCEL intervals are  almost equal, JEL confidence interval has shorter length.

\begin{table}[]
\centering
\caption{Empirical type 1 error for different values of $\nu$ }
\label{my-label}
\begin{tabular}{|c|c|c|c|c|ccccccccccc} \hline
                     & $n$    & {Exp(1)}      & {Pareto(1,2.5)}   & {Log normal(0,1)} \\ \hline
                    \multirow{5}{*}{$\nu=2$} & 25 &  0.098 & 0.087 & 0.123 \\
                     & 50  & 0.066& 0.068 & 0.109 \\
                     & 100 & 0.062 & 0.060 & 0.089  \\
                     & 200 & 0.056 & 0.056 & 0.066 \\
                     & 300 & 0.051 & 0.052& 0.060\\ \hline
\multirow{5}{*}{$\nu=3$} & 25  & 0.106 & 0.099 & 0.120  \\
                     & 50  & 0.068& 0.069 & 0.105 \\
                     & 100 & 0.062 & 0.064 & 0.085  \\
                     & 200 & 0.055 & 0.057 & 0.062 \\
                     & 300 & 0.051 & 0.052& 0.058\\ \hline
\multirow{5}{*}{$\nu=4$} & 25  & 0.990& 0.892 & 0.114 \\
                     & 50  & 0.066& 0.068 & 0.109 \\
                     & 100 & 0.062 & 0.061 & 0.080  \\
                     & 200 & 0.051 & 0.055 & 0.060 \\
                     & 300 & 0.051 & 0.052& 0.056\\ \hline
\multirow{5}{*}{$\nu=5$} & 25  & 0.951 & 0.872& 0.111 \\
                    & 50  & 0.064& 0.066 & 0.106 \\
                     & 100 & 0.060 & 0.058 & 0.079  \\
                     & 200 & 0.051 & 0.052 & 0.058 \\
                     & 300 & 0.050& 0.050& 0.053\\ \hline
\end{tabular}
\end{table}
 Next, we find the empirical type 1 error of the JEL based test and the result is reported in Table 4. We find the type 1 error rate  for $\nu=2,3,4,5$ when the samples are generated from standard exponential, Paeto with parameters $k=1$ and $a=2.5$ and standard log normal distributions. From Table 4, it is evident  that the empirical type 1 error  reaches the nominal value $\alpha=0.05$ as the sample size increases.

In Tables 5, 6 and 7 we report the  empirical power of the JEL based test  when the alternate hypothesis  is specified by the scenario  given below.
\\
1) $R_2=0.218 $, $R_3= 0.479$, $R_4=0.609 $, $R_5= 0.687$ ($X \sim $Exp(0.8))

 $R_2= 0.382$, $R_3=0.588 $, $R_4= 0.691$, $R_5=0.753 $ ($X \sim $Exp(0.9))

$R_2= 0.777$, $R_3=0.851 $, $R_4= 0.888$, $R_5=0.911 $ ($X \sim $Exp(1.5))
\\ 2) $R_2=0.148$, $R_3=0.181 $, $R_4= 0.200$, $R_5=0.210 $ ($X \sim $Pareto(1, 4))

$R_2=0.111 $, $R_3=0.142 $, $R_4=0.157$, $R_5=0.166 $ ($X \sim $Pareto(1, 5))

$R_2=0.052 $, $R_3=0.068 $, $R_4=0.076$, $R_5=0.081 $ ($X \sim $Pareto(1, 10))
\\ 3) $R_2=0.711 $, $R_3=0.836 $,$R_4=0.884$, $R_5=0.910$($X \sim $ Log normal(0, 1.5))

$R_2=0. 842$, $R_3=0.930 $, $R_4=0.958 $, $R_5=0.971$ ($X \sim $Log normal(0, 2))

$R_2=0.966 $, $R_3=0.991 $, $R_4=0.996 $, $R_5=0.998$ ($X \sim $Log normal(0, 3))

\noindent  From Tables 5, 6 and 7 it is clear that the proposed JEL test has good power even for small sample size in all the nine alternatives specified above.

\begin{table}[]\begin{small}
\centering
\caption{Empirical Power: Exponential Distribution}
\label{}
\begin{tabular}{|c|c|c|c|c|c|c} \hline
& $n$       &\quad $\theta=0.8$  \quad & \quad   {$\theta=0.9$} \quad &\quad$\theta=1.5$ \quad  \\ \hline
\multirow{5}{*}{$\nu=2$} & 25    & 0.739       &    0.742 &    0.740     \\
                     & 50  & 0.973      &   0.974  &    0.962    \\
                     & 100 & 0.995      &    0.999 &    0.984   \\
                     & 200 &  1.000     &   1.000   &   1.000   \\ \hline
\multirow{5}{*}{$\nu=3$} & 25   & 0.753       &   0.766   &      0.794 \\
                     & 50  &  0.966    &   0.979   &      0.894  \\
                     & 100 & 1.000    &   1.000 &   1.000     \\
                     & 200 &  1.000    &   1.000&   1.000     \\ \hline
\multirow{5}{*}{$\nu=4$} & 25     & 0.780      &      0.794 &      0.798    \\
                     & 50  &    0.980      &   0.986   &      0.990  \\
                     & 100 &  0.995       &   1.000  &   1.000    \\
                     & 200 &  1.000        &   1.000 &   1.000     \\ \hline
\multirow{5}{*}{$\nu=5$} & 25  & 0.786      &    0.808  &      0.804   \\
                     & 50  &  0.990      &   0.996 &      0.994  \\
                     & 100 &  1.000        &   1.000 &   1.000    \\
                     & 200 &  1.000      &   1.000   &   1.000   \\ \hline
\end{tabular}\end{small}
\end{table}

\begin{table}[]\begin{small}
\centering
\caption{Empirical Power: Pareto Distribution }
\label{}
\begin{tabular}{|c|c|c|c|c|c|c} \hline
& $n$       &$\alpha=4,\,k=1$   &    $\alpha=5,\,k=1$  &    $\alpha=10,\,k=1$  \\ \hline
\multirow{5}{*}{$\nu=2$} & 25    & 0.796       &   0.802  &   0.780    \\
                     & 50  & 0.930      &   0.944  &   0.934   \\
                     & 100 & 0.965      &    0.989 &    0.949    \\
                     & 200 &  0.997     &    1.000   &    0.992   \\ \hline
\multirow{5}{*}{$\nu=3$} & 25   & 0.753        &  0.766  &    0.808    \\
                     & 50  &  0.866    &    0.879  &    0.885     \\
                     & 100 & 0.999    &   1.000   &    1.000    \\
                     & 200 &  1.000    &   1.000   &    1.000   \\ \hline
\multirow{5}{*}{$\nu=4$} & 25     & 0.780      &      0.794   &    0.795    \\
                     & 50  &    0.838      &    0.846    &    0.850   \\
                     & 100 &  0.955       &    0.959   &    0.962 \\
                     & 200 &  0.999       &   1.000  &    1.000     \\ \hline
\multirow{5}{*}{$\nu=5$} & 25  & 0.706      &    0.708  &    0.717     \\
                     & 50  &  0.902      &   0.917  &    0.930    \\
                     & 100 &  0.962       &    1.000   &    1.000   \\
                     & 200 &  1.000      &   1.000   &    1.000    \\ \hline
\end{tabular}\end{small}
\end{table}


\begin{table}[]\begin{small}
\centering
\caption{Empirical  Power: Log Normal Distribution}
\label{}
\begin{tabular}{|c|c|c|c|c|c|c}\hline
& $n$       & $\mu=0,\,\sigma^2=1.5$      &    {$\mu=0,\,\sigma^2=2$ }  & {$\mu=0,\,\sigma^2=4$ }    \\ \hline
\multirow{5}{*}{$\nu=2$} & 25    & 0.576       &  0.582  &    0.579     \\
                     & 50  & 0.730      &    0.774  &    0.742     \\
                     & 100 & 0.935      &    0.959  &    0.940   \\
                     & 200 &  0.998     &   1.000   &   1.000    \\ \hline
\multirow{5}{*}{$\nu=3$} & 25   & 0.583       &    0.566   &  0.570   \\
                     & 50  &  0.766    &   0.763  &   0.759   \\
                     & 100 & 0.991    &   0.979  &   0.984    \\
                     & 200 &  0.998    &   1.000 &   1.000    \\ \hline
\multirow{5}{*}{$\nu=4$} & 25     & 0.590      &      0.594  &  0.590     \\
                     & 50  &    0.798      &   0.786    &  0.792   \\
                     & 100 &  0.955       &    0.950  &    0.953     \\
                     & 200 &  0.999       &   1.000   &    1.000    \\ \hline
\multirow{5}{*}{$\nu=5$} & 25  & 0.606      &    0.908  &  0.990     \\
                     & 50  &  0.806      &  0.990  &  0.990    \\
                     & 100 &  0.962       &    1.000  &    1.000    \\
                     & 200 &  1.000      &   1.000   &    1.000    \\ \hline
\end{tabular}\end{small}
\end{table}

\section{Application to real data}
We illustrate the proposed JEL based method using per capita personal income  data of the United States. The data is collected from U.S. Bureau of Economic Analysis and it is available  on $https://www.bea.gov$.
 The data illustrates quarter wise per capita personal income for the states in U.S. for the period 2013 to 2017 and is reported in dollar.  Relative S-Gini index for $v=3$  is calculated for each quarter and presented in Table 8. It can be noted that $R_3$ is slightly lower for the year 2016 and 2017. It suggest lesser inequality for that period.

  We find the confidence interval for $R_3$ using jackknife empirical likelihood method and the result is  reported in Table 8. From Table 8 we can see that the average length of the intervals is higher for the year 2015 to 2017. This explains that the data for these years have more variability compared to that of previous years, 2013 and 2014,  across the states of US .
\begin{singlespace}
\begin{table}[h]
\centering
\caption{Per capita personal income:  95\% confidence interval for $R_3$}
\label{my-label}
\begin{tabular}{|l|l|l|l|l|}
\hline
Quarter    & $\widehat R_3$      & Lower limits                                                                            & Upper limits  & Average length                                                                                                                                                          \\ \hline
Q1 2013   &0.1509 &0.1071  &0.2152 &0.1081  \\ \hline
Q2 2013  & 0.1510  &0.1069 &0.2103& 0.1034\\ \hline
Q3 2013     &0.1508 &0.1070 & 0.2044 &0.0974\\ \hline
Q4 2013    & 0.1519&0.1074 &0.2056 &0.0982 \\ \hline
Q1 2014   & 0.1514&0.1078  &0.2264&0.1186   \\ \hline
Q2 2014  &0.1512   &0.1074 &0.2241 &0.1167\\ \hline
Q3 2014     &0.1515 &0.1068 & 0.2250 &0.1182\\ \hline
Q4 2014    & 0.1510&0.1060 &  0.2205&0.1145\\ \hline
Q1 2015   & 0.1506 &  0.1057 &0.2750& 0.1693  \\ \hline
Q2 2015  & 0.1509  &0.1059 &0.2577& 0.1518\\ \hline
Q3 2015     & 0.1503&0.1055 &0.2553& 0.1498\\ \hline
Q4 2015    &0.1501 &0.1055 &0.2518 & 0.1463\\ \hline
Q1 2016   & 0.1501&  0.1050&0.2749& 0.1699 \\ \hline
Q2 2016  &  0.1503 & 0.1050&0.2761& 0.1711\\ \hline
Q3 2016     & 0.1506&0.1056 &0.2745 &0.1689 \\ \hline
Q4 2016    & 0.1502&0.1052 &0.2759 &0.1707 \\ \hline
Q1 2017   & 0.1490&0.1034  &0.2757  &0.1723 \\ \hline
Q2 2017  & 0.1500  &0.1048 &0.2689 &0.1641\\ \hline
Q3 2017     &0.1503 & 0.1055& 0.2749 &0.1694\\ \hline
Q4 2017    & 0.1506&0.1057 & 0.2750 &0.1693 \\ \hline
\end{tabular}
\end{table}
\end{singlespace}

\newpage
\section{Conclusion}
\par Gini index are generalised into many families of income inequality measures and  S-Gini indices is one among them. S-Gini indices are extensively used to study income inequality and to evaluate the performance of stocks in finance. We obtained simple non-parametric estimator for S-Gini indices and proved the asymptotic properties of the proposed estimator using the asymptotic theory  of  U-statistics.  We derived the limiting distribution of empirical log likelihood ratio as well as jackknife empirical log likelihood ratio for relative S-Gini indices.  The simulation study shows that JEL based confidence interval performs better than that of bootstrap-t and bootstrap calibrated empirical likelihood confidence intervals in terms of coverage probability and average length. The simulation study also shows that the proposed JEL based test has well controlled error rate and have good power for different alternatives. Finally we illustrate our method using per capita personal income  data of the United States.

\begin{thebibliography}{22}
\providecommand{\natexlab}[1]{#1}
\providecommand{\url}[1]{\texttt{#1}}
\expandafter\ifx\csname urlstyle\endcsname\relax
  \providecommand{\doi}[1]{doi: #1}\else
  \providecommand{\doi}{doi: \begingroup \urlstyle{rm}\Url}\fi


\bibitem{berret}
Barret, G. F. and Donald, S. G. (2009). Statistical Inference with Generalized Gini Indices of Inequality,
Poverty, and Welfare,
\newblock \emph{Journal of Business \& Economic Statistics}, 27\penalty0 (1), \penalty0 1--17.



\bibitem{demuynck2012almost}
Demuynck, T. (2012),
\newblock An (almost) unbiased estimator for the {S-Gini} index,
\newblock \emph{The Journal of Economic Inequality}, 10\penalty0 (1), \penalty0
  109--126.

\bibitem{donaldson1980single}
Donaldson, D. and Weymark, J. A. (1980),
\newblock A single-parameter generalization of the {Gini} indices of
  inequality,
\newblock \emph{Journal of Economic Theory}, 22\penalty0 (1), \penalty0 67--86.


\bibitem{formby2001sen}
Formby, J.,  Kim, H. and Zheng, B. (2001),
\newblock Sen measures of poverty in the united states: cash versus
  comprehensive incomes in the 1990s,
\newblock \emph{Pacific Economic Review}, 6\penalty0(2), \penalty0 193--210.

%


\bibitem{palmitesta2006asymptotic}
Giorgi, G. M., Palmitesta, P. and Provasi, C. (2006),
\newblock Asymptotic and bootstrap inference for the generalized gini indices,
\newblock \emph{Metron}, 64\penalty0 (1), \penalty0 107--124.


\bibitem{hall92bootstrap}
Hall, P. (1992),
\newblock \emph{The bootstrap and Edgeworth expansion},
\newblock Springer Science \& Business Media.

\bibitem{hoeffding1948class}
Hoeffding, W. (1948),
\newblock A class of statistics with asymptotically normal distribution,
\newblock \emph{The annals of mathematical statistics}, 19\penalty0
  (3), \penalty0 293--325.

\bibitem{jing2009jackknife}
Jing, B. Y, Yuan, J. and Zhou, W. (2009),
\newblock Jackknife empirical likelihood,
\newblock \emph{Journal of the American Statistical Association}, 104\penalty0
  (487), \penalty0 1224--1232.


\bibitem{lehmann1951consistency}
Lehmann, E. L. (1951),
\newblock Consistency and unbiasedness of certain non-parametric tests,
\newblock \emph{The Annals of Mathematical Statistics}, 22\penalty0
  (2), \penalty0 165--179.
\bibitem{}Luo, S. and Qin, G. (2018). Jackknife empirical likelihood-based inferences for Lorenz curve with kernel smoothing,
\newblock \emph{Communications in Statistics - Theory and Methods}, online first, https://doi.org/10.1080/03610926.2017.1417426

\bibitem{}
Lv, X., Zhang, G., Xu, X., and Li, Q. (2017). Bootstrap-calibrated empirical likelihood confidence intervals for the difference between two Gini indexes.
\newblock \emph{The Journal of Economic Inequality}, 15, 195-216.



\bibitem{owen1988empirical}
Owen, A. B. (1988),
\newblock Empirical likelihood ratio confidence intervals for a single
  functional,
\newblock \emph{Biometrika}, 75\penalty0 (2), \penalty0 237--249.

\bibitem{owen1990empirical}
Owen, A. B. (1990),
\newblock Empirical likelihood ratio confidence regions,
\newblock \emph{The Annals of Statistics}, 18\penalty0 (1), \penalty0 90--120.


\bibitem{peng2011empirical}
Peng, L. (2011),
\newblock Empirical likelihood methods for the {Gini} index,
\newblock \emph{Australian \& New Zealand Journal of Statistics}, 53\penalty0
  (2), \penalty0 131--139.

\bibitem{qin2010empirical}
Qin, Y.,  Rao, J. N. K. and Wu, C. (2010),
\newblock Empirical likelihood confidence intervals for the {Gini} measure of
  income inequality,
\newblock \emph{Economic Modelling}, 27\penalty0 (6), \penalty0 1429--1435.


     \bibitem{QY2013}
Qin, G.,  Yang, B. and Hall, N. E. B. (2013), Empirical likelihood based inferences for Lorenz curve, \textit{Annals of Institute of Statistical Mathematics}, 65(1), 1--21.

\bibitem{} Sang, Y, Dang, X. and Zhao, Y. (2019). Jackknife empirical likelihood methods for Gini correlations and their equality testing, {\em Journal of Statistical Planning and Inference}, 199, 45-59.


\bibitem{} Wang, D., Zhao, Y. and  Gilmore, D.W. (2016).  Jackknife empirical likelihood confidence interval for the Gini index, {\em Statistics \& Probability Letters }, 110(1), 289-295.

\bibitem{} Wang, D. and Zhao, Y. (2016). Jackknife empirical likelihood for comparing two Gini indices, {\em Canadian Journal of Statistics},  44(1), 102-119.


\bibitem{xu2000}
Xu, K. (2000),
\newblock Inference for generalized gini indices using the iterated-bootstrap
  method,
\newblock \emph{Journal of Business \& Economic Statistics}, 18\penalty0
  (2), \penalty0 223--227.

\bibitem{xu2007u}
Xu, K. (2007),
\newblock U-statistics and their asymptotic results for some inequality and
  poverty measures,
\newblock \emph{Econometric Reviews}, 26\penalty0 (5), \penalty0 567--577.

\bibitem{yitzhaki1983extension}
Yitzhaki, S. (1983),
\newblock On an extension of the {Gini} inequality index.
\newblock \emph{International economic review}, 24\penalty0 (3), \penalty0
  617--628.

 \bibitem{yitzhaki2013}
 Yitzhaki, S. and Schechtman, E. (2013),
 \newblock \emph{ The Gini Methodology: A Primer on a Statistical Methodology},
 \newblock{Springer.}

 \bibitem{zitikis2003}
Zitikis, R. (2003),
\newblock Asymptotic estimation of {E--Gini} index,
\newblock \emph{Econometric Theory}, 19\penalty0
 (4), \penalty0 587--601.

\bibitem{zitikis2002asymptotic}
Zitikis, R. and Gastwirth, J. L.(2002),
\newblock The asymptotic distribution of the {S--Gini} index,
\newblock \emph{Australian \& New Zealand Journal of Statistics}, 44\penalty0
 (4), \penalty0 439--446.

\end{thebibliography}
\end{document}